\documentclass[a4paper,USenglish,thm-restate]{lipics-v2021}
\usepackage{amsmath}% AMS math
\usepackage{amssymb}% special AMS math symbols
\usepackage{amsthm}% AMS theorem/proof
\usepackage{amsfonts}% more special AMS math symbols
\usepackage{graphicx}% graphics (includegraphics)
\usepackage{cite}% sort bibliographical entries in citations
\usepackage{hyperref}% add hyper-references
\hypersetup{colorlinks=true,linkcolor=[rgb]{0.75,0,0},citecolor=[rgb]{0,0,0.75}}
\hypersetup{colorlinks=true,linkcolor=[rgb]{0.75,0,0},citecolor=[rgb]{0,0,0.75}}
\usepackage{tikz}
\usetikzlibrary{arrows,calc,patterns}
\usepackage{chngcntr}
\usepackage{apptools}
\AtAppendix{\counterwithin{apxlem}{section}}

\usepackage{thmtools}
\usepackage{thm-restate}
\newcommand{\RE}{\mathbb{R}}

\newcommand{\eps}{\varepsilon}
\newcommand{\bd}{\partial}

\newcommand{\LFS}{\phi}
\newcommand{\capsule}{\mathcal{C}}

\newcommand{\inv}[1]{\frac{1}{#1}}

\newcommand{\Hess}{\nabla^2}
\newcommand{\Transpose}{\intercal}

\newcommand{\SP}{\kern+1pt}

\DeclareMathOperator{\diam}{diam}
\DeclareMathOperator{\dist}{dist}

\DeclareMathOperator{\vol}{vol}

\DeclareMathOperator{\Cyl}{Cyl}
\nolinenumbers

\newcommand{\mytitle}{Approximate Nearest-Neighbor Search for Line Segments}
\title{\mytitle}
\titlerunning{\mytitle}%optional, please use if title is longer than one line

\author{Ahmed Abdelkader}{Oden Institute for Computational Engineering and Sciences \\ The University of Texas at Austin, Austin, USA}{akader@cs.umd.edu}{https://orcid.org/0000-0002-6749-1807}{}%mandatory, please use full name; only 1 author per \author macro; first two parameters are mandatory, other parameters can be empty.

\author{David M. Mount}{Department of Computer Science and Institute of Advanced Computer Studies \\ University of Maryland, College Park MD, USA}{mount@umd.edu}{https://orcid.org/0000-0002-3290-8932}{}

\authorrunning{A.\ Abdelkader and D.\ M.\ Mount.}%mandatory. First: Use abbreviated first/middle names. Second (only in severe cases): Use first author plus 'et. al.'

\Copyright{Ahmed Abdelkader and David M. Mount}%mandatory, please use full first names. LIPIcs license is "CC-BY";  http://creativecommons.org/licenses/by/3.0/


%mandatory

\funding{Supported by NSF grant CCF-1618866 and an Ann G.\ Wylie Dissertation Fellowship.}%optional, to capture a funding statement, which applies to all authors. Please enter author specific funding statements as fifth argument of the \author macro.

\acknowledgements{The first author's work was conducted in part at the University of Maryland.}%optional

\begin{document}

\maketitle

%-----------------------------------------------------------------------
\begin{abstract}
Approximate nearest-neighbor search is a fundamental algorithmic problem that continues to inspire study due its essential role in numerous contexts. In contrast to most prior work, which has focused on point sets, we consider nearest-neighbor queries against a set of line segments in $\RE^d$, for constant dimension $d$. Given a set $S$ of $n$ disjoint line segments in $\RE^d$ and an error parameter $\varepsilon > 0$, the objective is to build a data structure such that for any query point $q$, it is possible to return a line segment whose Euclidean distance from $q$ is at most $(1+\varepsilon)$ times the distance from $q$ to its nearest line segment. We present a data structure for this problem with storage $O((n^2/\varepsilon^{d}) \log (\Delta/\varepsilon))$ and query time $O(\log (\max(n,\Delta)/\varepsilon))$, where $\Delta$ is the spread of the set of segments $S$. Our approach is based on a covering of space by anisotropic elements, which align themselves according to the orientations of nearby segments.
\end{abstract}
%-----------------------------------------------------------------------

%=======================================================================
\section{Introduction} \label{sec:intro}
%=======================================================================

Proximity queries are essential building blocks in many important algorithms with numerous applications \cite{FPS96, CoH67, DuH73, CoS93, GeG91, FSN95, DDF90, DeW82}. A primary example is \emph{nearest-neighbor searching}, where a given set of $n$ points $P$ in $\RE^d$ is preprocessed into a data structure so that queries can be answered efficiently. As the complexity bounds for such query problems grow very rapidly as the dimension increases, either in terms of query time or space, most research has focused on approximate solutions. There has been a great deal of work on approximate proximity searching in spaces of very high dimension \cite{HIM12, Ind04, CLM08, ANNRW18a} and in general metric spaces \cite{KrL04, KrL05, HaM06, ANNRW18a}. Nonetheless, there are many important applications that naturally reside in real spaces of relatively low dimensions.

In this paper, we consider approximate nearest-neighbor searching for a query point against a discrete set of line segments in $\RE^d$, where $d$ is a fixed constant. We are given a set $S$ of $n$ disjoint line segments in $\RE^d$. The distance from any point $q \in \RE^d$ to a segment $s$, denoted $\dist(q,s)$, is the minimum Euclidean distance between $q$ and any point of $s$. For $\eps > 0$, a segment $s' \in S$ is an \emph{$\eps$-approximate nearest neighbor} (\emph{$\eps$-ANN}) of $q$ if $\dist(q,s')$ is within a factor of $1+\eps$ of the distance to $q$'s closest segment in $S$. Given $S$ and $\eps > 0$, the objective is to construct a data structure so that given any $q \in \RE^d$, it is possible to compute an $\eps$-ANN of $q$ efficiently. We refer to this problem as \emph{segment-ANN}.

Clearly, nearest-neighbor searching with respect to segments is at least as hard as for point sets, and there are quadratic worst-case lower bounds in the exact and approximate settings \cite{Har01, Aro02}. The difficulty of the problem can be appreciated by considering the increased complexity of the Voronoi diagram of a set of lines or line segments, which is not fully understood to date~\cite{aurenhammer_et_al:LIPIcs:2017:8215, barequet_et_al:LIPIcs:2019:11558}. Recall that such Voronoi diagrams consist of cells bounded by hyperplanes and algebraic surfaces of constant degree, and hence are generally nonconvex; see~\cite{SuI94,karavelas2004robust} for the computation of Voronoi diagrams of line segments in $\RE^2$. A set of just three straight lines in $\RE^3$ suffices to induce a highly intricate Voronoi diagram~\cite{Everett2009}. Better bounds are known in restricted scenarios, for example, by bounding the number of orientations~\cite{KoS03,emiris2010approximate} or working with a polyhedral distance function~\cite{Koltun2004,CKS98}.

Our work follows in this tradition with the main motivation of developing a better understanding of proximity searching among more complex objects than discrete sets of points. We are particularly interested in distance functions whose rate of change is much larger in some directions compared to others. This sort of behavior is characterized by the notion of \emph{anisotropy}, which can be defined for smooth convex functions as the ratio of the largest to the smallest eigenvalues of the Hessian matrix at the point in question. Line segments are perhaps the simplest objects inducing such distance functions. This type of proximity searching against linear and affine subspaces has recently been applied to problems in pattern recognition~\cite{BHZ11,WAW13} and active learning~\cite{VJG14}. A notion of \textit{direction-sensitive distances} has been studied in the plane~\cite{aichholzer1999skew}; see also the work on \textit{boat-sail distances}~\cite{Sug11}.

In this paper we present a new data structure for segment-ANN. Our data structure is an AVD-style data structure \cite{Har01, ArM02, AVD-JACM, AFM17a}. By this we mean that it employs a hierarchical subdivision of space (a covering in our case) by elements of constant complexity (ellipsoids in our case). At the leaf level of the hierarchy, each element stores a \emph{representative segment} of $S$ that is an $\eps$-ANN for any query point lying within the element. Queries are answered by a simple descent through the hierarchy, reporting the representative of the leaf-level element. Up to now, AVD structures have relied on quadtree-based subdivisions. A novel feature of our approach is that the elements are anisotropic, where their shapes are sensitive to the local distribution of segments. The advantages of such an approach are illustrated intuitively in Figure~\ref{quadtree-simple.fig}. Our approach is inspired by recent progress on the use of anisotropic covering elements based on Macbeath regions \cite{Mac50, Bar00} used in convex approximation \cite{AFM17b, AFM17c, AFM18a, AFM18b, AbM18}.

%-----------------------------------------------------------------------
\begin{figure}[htbp]
  \centerline{\includegraphics[scale=0.40]{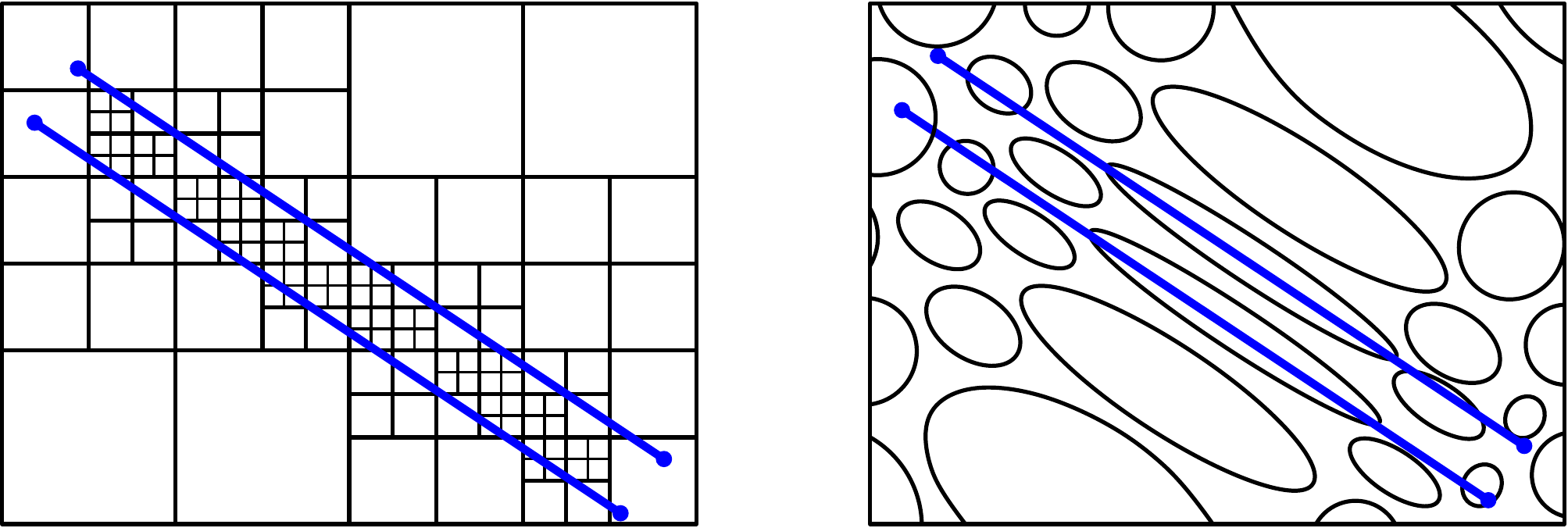}}
  \caption{Approximation using isotropic (quadtree) elements compared to anisotropic elements.}
  \label{quadtree-simple.fig}
\end{figure}
%-----------------------------------------------------------------------

Our input consists of a set $S$ of $n$ pairwise disjoint line segments in $\RE^d$. Define the \emph{spread}, denoted $\Delta(S)$ to be $\diam(S)/\delta_{\min}(S)$, where $\diam(S)$ is the diameter of the set $S$ (the maximum distance between any two points lying on these segments), and $\delta_{\min}$ is the minimum distance between any two segments. Since $S$ will be fixed throughout, we will just refer to this as $\Delta$. Here is our main result.

%-----------------------------------------------------------------------
\begin{theorem} \label{ann-main.thm}
Given a set $S$ of $n$ disjoint line segments in $\RE^d$ of spread $\Delta$ and $\eps > 0$, there exists a data structure that can answer $\eps$-ANN queries in time $O(\log (\max(n,\Delta)/\eps))$ using $O((n^2/\eps^d) \log \frac{\Delta}{\eps})$ storage.
\end{theorem}
%-----------------------------------------------------------------------

Note that because the line segments are disjoint and the dimension is constant, $n = O(\Delta^d)$, the query time bound can be simplified to $O(\log \frac{\Delta}{\eps})$. 

The most closely related works to ours are segment-ANN data structures by Mahabadi~\cite{Mah15} and Agarwal, Rubin, and Sharir \cite{ARS17}. Mahabadi's solution is based on reducing segment-ANN to point-ANN through a combination of reductions. These reductions produce $n^{O(1)}$ point-ANN modules, where each module involves $O(n/\eps^{O(1)})$ points. The space bounds obtained are inferior to ours in terms of $n$ and $\eps$, and while individual modules can be solved within the AVD model, the overall data structure is not in this model. Agarwal, Rubin, and Sharir\cite{ARS17} consider the more general problem of ANN queries against $k$-flats in $\RE^d$. As in our case, $d$ is assumed to be constant. They solve the problem by approximating the Euclidean ball with a polyhedron distance function of complexity $1/\eps^{O(1)}$ \cite{Boissonnat1998}, and they show that it is possible to compute nearest neighbors exactly among $k$-flats with respect to the induced polyhedral distance function through the use of multi-level partition trees. In the case of line segments, their approach provides polylogarithmic query time with $n^2 (\log (n)/\eps)^{O(1)}$ storage, but the approach makes critical use of the fact that the objects are (infinite) flats. As with Mahabadi's result, there is no dependence on the spread. There are also works that consider the dual problem, where the data set consists of points and the query is a $k$-flat~\cite{AIKN09,MNSS15,ARS17}.

Our data structure has a number of notable features. First, it is in the AVD model (which partially answers an open problem posed by Agarwal, Rubin, and Sharir \cite{ARS17}. The query algorithm is almost trivial, involving a descent through a rooted directed acyclic graph (DAG) of constant degree. The decision of which neighbor to visit next is just a membership test for an ellipsoid. By abandoning the quadtree-based approaches used in prior AVD solutions, we demonstrate how to exploit the anisotropic nature of the nearest-neighbor distance function to obtain a space-efficient hierarchical spatial decomposition. 

The remainder of the paper is organized as follows. Section~\ref{sec:anisotropy} formalizes the notion of anisotropy by examining the differential properties of the distance to the segments. Section~\ref{sec:capsules} introduces the notion of a \emph{capsule}, the basic shape upon which our data structure is built, and introduces the relevant properties of these objects. Section~\ref{data-struct.sec} presents our ANN search structure, and Section~\ref{sec:storage} analyzes its storage complexity. For completeness, in Appendix~\ref{lower-bound.sec} we include a proof of the quadratic lower bound in our anisotropic generalization of the AVD.

%=======================================================================
\section{Exposing Anisotropy} \label{sec:anisotropy}
%=======================================================================

In this section, we formally characterize how the distance function associated by a set of line segments naturally induces a Riemannian metric whose metric tensor is anisotropic; see Figure~\ref{ellipsoids.fig}. This characterization underpins the design of our data structure which draw inspiration from classical constructions in convex optimization. For the sake of efficiency, the construction of our data structure will be based on a simpler approach, and this section may be skimmed without hampering the understanding of the material that follows.

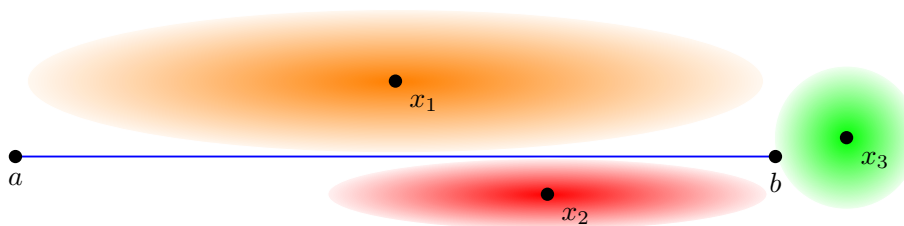
\begin{figure}[htbp]
\centering
\begin{tikzpicture}[
dot/.style = {circle, fill, minimum size=#1,
              inner sep=0pt, outer sep=0pt},
dot/.default = 6pt
                    ]
    
    \draw [line width=0.25mm, blue ] (0,0) -- (10, 0);
    \draw [white,inner color=orange,outer color=orange!5] (5,1) ellipse ({0.95*sqrt(26)} and 0.95*1);
    \draw [white,inner color=red,outer color=red!5] (7,-0.5) ellipse ({0.95*sqrt(9.25)} and 0.95*0.5);
    \draw [white,inner color=green,outer color=green!5] ({10+sqrt(0.875)},0.25) ellipse (0.95*1 and 0.95*1);

    \node[dot=5pt, label=below:{\large $a$}] at (0,0) {};
    \node[dot=5pt, label=below:{\large $b$}] at (10,0) {};
    \node[dot=5pt, label=below right:{\large $x_1$}] at (5,1) {};
    \node[dot=5pt, label=below right:{\large $x_2$}] at (7,-0.5) {};
    \node[dot=5pt, label=below right:{\large $x_3$}] at ({10+sqrt(0.875)},0.25) {};
\end{tikzpicture}
  \caption{Demonstrating the local tensors induced by a segment $\overline{ab}$ as defined in Equation~\ref{eq:local_ellipsoid}.}
  \label{ellipsoids.fig}
\end{figure}
%-----------------------------------------------------------------------

Given a set of pairwise-disjoint segments $S = \{s_1, \dots, s_n\}$, denote by $\ell_i$ the line supporting $s_i = \overline{a_i b_i}$ parallel to the unit vector $v_i$. We define the distance functions at any $x \in \RE^d$ as
\begin{equation}\label{eq:distance_to_segment}
    D_i(x) ~ = ~ \begin{cases}
           D^\ell_i(x), & \text{if } x^\perp \in \text{int}(s_i), \\
           D^\bullet_i(x), & \text{otherwise},
          \end{cases}
\end{equation}
where $D^\ell_i$ is half the squared distance to the line $\ell_i$, and $D^\bullet_i(x) = \min\{D^{a_i}(x), D^{b_i}(x)\}$ with $D^{a_i}$ and $D^{b_i}$ being half the squared distance to the endpoints $a_i$ and $b_i$, respectively, $x^\perp$ the projection of $x$ onto $\ell_i$, and $\text{int}(s_i)$ the interior of $s_i$. As is common for similar definitions, we work with squared distances and introduce the $\frac{1}{2}$ factor to simplify the resulting derivatives.

For every $x \in \RE^d$, we seek a definition of a \emph{local tensor} to effectively consolidate the two cases in the definition of $D_i(x)$ per Equation~\ref{eq:distance_to_segment}. Using such local tensors, we can define a \emph{local descriptor}, e.g., an ellipsoid, whose shape describes the rate of change of the distance function $D_i$ in the neighborhood of $x$. We achieve this by first examining the Hessian of the distance functions defining $D_i(x)$ for each segment in isolation. Then, we consider the consolidation of all distance functions as needed for nearest-neighbor searching.

%=======================================================================
\subsection{Distance Hessians} \label{sec:derivation}
%=======================================================================

For a fixed point $p \in \RE^d$, the associated distance takes the form
\begin{equation} \label{eq:point_hess}
    D^p(x) 
        ~ = ~ \frac{1}{2}\|x - p\|^2 
        ~ = ~ \frac{1}{2} \sum_{i=1}^d (x_i - p_i)^2, \quad\text{for which $\Hess D^p = I$,}
\end{equation}
where $\Hess$ denotes the function's Hessian and $I$ is the identity matrix. For a fixed line $\ell = \{p + t v \mid t \in \RE\}$, with $p,v \in \RE^d$ and $\|v\|=1$, the distance takes the form
\begin{align*}
    D^\ell(x) 
          ~ = ~ \frac{1}{2}\|x - x^\perp\|^2
        & ~ = ~ \frac{1}{2}\sum_{i=1}^d \left((x_i-p_i) - \langle x-p, v \rangle v_i\right)^2,
\end{align*}
where $x^\perp$ is the projection of $x$ onto $\ell$. We proceed to compute the Hessian $\Hess D^\ell$ as follows.
\[
    \frac{\partial D^\ell}{\partial x_k} 
        ~ = ~ (x_k - p_k) - \langle x - p, v \rangle v_k, \qquad
    \frac{\partial^2 D^\ell}{\partial x_k^2} 
        ~ = ~ 1-v_k^2, \qquad
    \frac{\partial^2 D^\ell}{\partial x_k\partial x_n}
        ~ = ~ -v_k v_n,
\]
\begin{equation} \label{eq:line_hess}
    \Hess D^\ell 
        ~ = ~ I - v v^\Transpose.
\end{equation}
It is easy to verify that $v$ is an eigenvector of $\Hess D^\ell$ with eigenvalue $0$. Letting $T$ be any rotation matrix such that $T v = [1, 0, \dots, 0]^\Transpose$, we obtain
% \begin{equation}\label{eq:canonical_hessian} (Currently not using the label
\[
    \Hess (D^\ell \circ T) 
        ~ = ~ \begin{bmatrix}
            0 & 0 & \dots & 0 \\
            0 & 1 & \dots & 0 \\
            \vdots & \vdots & \ddots & \vdots \\
            0 & 0 & \dots & 1
        \end{bmatrix}.
\]
Noting that $\Hess D^\ell = T^{-1}\Hess(D^\ell \circ T)T^{-1}$ and that eigenvalues are invariant under change of basis, the remaining eigenvalues of $\Hess D^\ell$ are all equal to $1$, where the corresponding eigenvectors can be chosen as any basis of the subspace orthogonal to $v$.  This form of the Hessian reflects the constancy of the distance along trajectories parallel to the line.

%=======================================================================
\subsection{Local Tensors and Ellipsoids} \label{sec:tensors}
%=======================================================================
Per the previous subsection, the Hessian $\Hess D^\ell_i$, of the distance to a line $\ell_i$, is rank-deficient with $v_i$ an eigenvector with eigenvalue $0$ and all remaining eigenvalues equal to $1$. We remedy this deficiency by defining the local tensor as
\begin{equation}\label{eq:tensor}
    \mathbb{H}_i(x) 
        ~ = ~ \frac{1}{D_i(x)} \Hess D^\ell_i + \frac{1}{D^\bullet_i(x)}v_i v_i^\Transpose.
\end{equation}
By construction, the local tensor $\mathbb{H}_i(x)$ has a single eigenvalue equal to $1/D^\bullet_i(x)$ with all remaining $d-1$ eigenvalues equal to $1/D_i(x)$. The anisotropy of the distance function $D_i$ reflected by this local tensor at $x$ is equal to the ratio of the maximum of $D^\bullet_i(x)$ and $D_i(x)$ to their minimum. As $x$ moves along any smooth trajectory, this anisotropy varies continuously between of $1$ and $\infty$. We use the tensor $\mathbb{H}_i(x)$ to define the ellipsoid
\begin{equation}\label{eq:local_ellipsoid}
    \mathcal{E}_i(x)
        ~ = ~ \left\{ y \in \RE^d ~\bigg|~ \frac{1}{2} (y - x)^\Transpose \mathbb{H}_i(x) (y - x) \leq 1 \right\}.
\end{equation}
(See Figure~\ref{ellipsoids.fig} for examples.) Observe that as $D^\bullet_i$ becomes larger, the eigenvalue associated with $v_i$ becomes smaller, and the ellipsoid $\mathcal{E}_i(x)$ extends further in the direction of $v_i$. On the other hand, as $D^\ell_i(x)$ approaches $D^\bullet_i(x)$, $\mathbb{H}_i(x)$ approaches a scaled identity matrix, and the ellipsoid $\mathcal{E}_i(x)$ becomes more spherical.

One approach to account for the influence of all $n$ segments is to define a \emph{blended tensor} at every $x \in \RE^d$, $\mathbb{H}(x) = \sum_{i=1}^n \mathbb{H}_i(x)$ along with an induced \emph{local norm}%
\footnote{See~\cite{NeT02, Nar16} for related derivations of Riemannian metrics from local tensors.}
and a corresponding \textit{local ellipsoid} acting as a \emph{metric ball} at $x$
\begin{equation}
    \widetilde{\mathcal{E}}(x) 
        ~ = ~ \left\{y \in \RE^d \mid \|y - x\|_x^2 \leq 1\right\}, \quad\text{where}\quad
    \|y - x\|^2_x 
        ~ = ~ \frac{1}{2} (y-x)^\Transpose \mathbb{H}(x) (y - x).
\end{equation}
Alternatively, we may directly bound the relative change of \emph{all} distance functions in the neighborhood of $x$ by restricting attention to the cell
\footnote{This can be seen by recognizing $\mathbb{H}_i(x)$ as the Hessian of a closely related function derived from $D_i$, and writing its Taylor expansion about $x$ for points within $\widetilde{\mathcal{E}}(x)$.}
\begin{equation}
    \widehat{\mathcal{E}}(x) 
        ~ = ~ \bigcap_{i=1}^n \mathcal{E}_i(x).
\end{equation}
The next lemma formalizes the relationship between the ellipsoids $\widetilde{\mathcal{E}}(x)$ and the cells $\widehat{\mathcal{E}}(x)$.%
\footnote{Readers familiar with the \emph{Dikin ellipsoid} from convex optimization will recognize the similarities with the local ellipsoids $\widetilde{\mathcal{E}}(x)$. It is well-known that Macbeath regions and Dikin ellipsoids are related by a similar inclusion as in Lemma~\ref{capsule-dikin.lem}: for a polytope $K$ defined as the intersection of $m$ halfspaces and a point $x \in K$, the Macbeath region $K \cap (2x - K)$ contains the Dikin ellipsoid at $x$ and is contained in its $\sqrt{m}$ expansion; see, e.g., \cite{SaV16, ATV20}. While Dikin ellipsoids are derived from barrier functions~\cite{Wright97,Nesterov14}, we derive our ellipsoids from the Euclidean distance functions.}

%-----------------------------------------------------------------------
\begin{lemma} \label{capsule-dikin.lem}
For any set of $n$ segments and any point $x \in \RE^d$, we have the inclusions
\begin{equation*} \label{eq:dikin}
    \widetilde{\mathcal{E}}(x) 
        ~ \subseteq ~ \widehat{\mathcal{E}}(x) 
        ~ \subseteq ~ \widetilde{\mathcal{E}}^{\sqrt{n}}(x), \quad\text{where the superscript denotes the central scaling about $x$.}
\end{equation*}
\end{lemma}
%-----------------------------------------------------------------------

%-----------------------------------------------------------------------
\begin{proof}
The first inclusion is immediate. For the second inclusion, observe that any $y \in \widehat{\mathcal{E}}(x)$ satisfies $\max_i \frac{1}{2} (y-x)^\Transpose \mathbb{H}_i(x) (y-x) \leq 1$. Hence, $\|y - x\|_x^2 \leq n$, implying $y \in \widetilde{\mathcal{E}}^{\sqrt{n}}(x)$.
\end{proof}
%-----------------------------------------------------------------------

Unfortunately, each of those two approaches has its own drawbacks. While the blended tensors $\mathbb{H}(x)$ are easier to compute, the corresponding ellipsoids $\widetilde{\mathcal{E}}(x)$ are unnecessarily small. On the other hand, the cells $\widehat{\mathcal{E}}(x)$ can retain a suitable size but are difficult to construct. This motivates an alternative, and more geometric, definition of a more efficient shape primitive.

%=======================================================================
\section{Anisotropic Space Covers} \label{sec:capsules}
%=======================================================================

Building upon the derivations in the previous section, we propose a simple primitive shape for constructing a hierarchical space covering, which will be amenable to computation and analysis. Recall that $S$ is a set of $n$ disjoint line segments in $\RE^d$, each defined by its two endpoints. Fix a segment $s = \overline{ab} \in S$, and let $r > 0$ be a given distance parameter, to be defined later. Recall that for any $x \in \RE^d$, its distance to segment $s$ is denoted by $\dist(x, s)$. 

For any point $x \in \RE^d$, we define a convex and centrally-symmetric \text{subregion} centered about $x$, called a \emph{capsule} and denote by $\capsule_{s}(x, r)$. If the closest point to $x$ on $s$ is an endpoint, then $\capsule_{s}(x, r)$ is simply the ball centered at $x$ with radius $\max(r, \dist(x, s))$. Otherwise, $\capsule_{s}(x, r)$ is defined as follows. First, construct the infinite cylinder of radius $\max(r, \dist(x, s))$ with axis parallel to $s$ and passing through $x$. Consider a ball centered at $x$ whose radius is $\max(r, \min(\|x - a\|, \|x - b\|)$, where $a$ and $b$ are the endpoints of $s$. The capsule is the intersection of this cylinder and ball (see Figure~\ref{capsules.fig}(a) and (b)).

%-----------------------------------------------------------------------
\begin{figure}[htbp]
  \centerline{\includegraphics[scale=0.40]{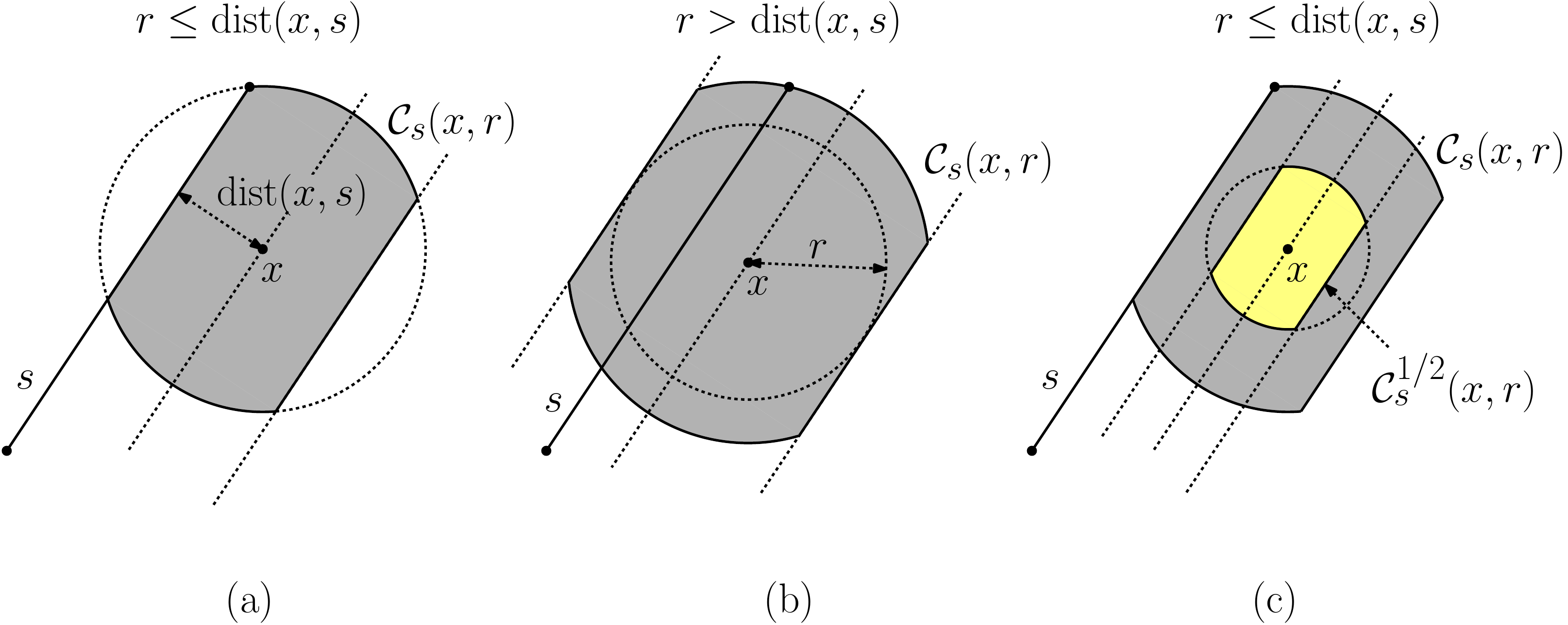}}
	\caption{(a) The capsule $\capsule_{s}(x, r)$ at $x$ for segment $s$ for $r \leq \dist(x,s)$, (b) for $r > \dist(x,s)$, and (c) the shrunken capsule $\capsule^{1/2}_{s}(x, r)$ for $r \leq \dist(x,s)$.}
  \label{capsules.fig}
\end{figure}
%-----------------------------------------------------------------------

We define the \emph{capsule} associated with $x$ for the set $S$ of all segments as $\capsule_S(x, r) = \bigcap_{s \in S} \capsule_{s}(x, r)$. Since $S$ will be fixed throughout, we will omit this subscript henceforth. Clearly, $\capsule(x, r)$ is also convex and centrally symmetric about $x$. We start by showing that capsules are closely related to the Hessian-based ellipsoids defined in Eq.~\eqref{eq:local_ellipsoid}.

%-----------------------------------------------------------------------
\begin{restatable}{lemma}{capsuleTensor}
For all $x \in \RE^d$ and any $r \leq \min_i \dist(x, s_i)$, 
\begin{equation}
    \widehat{\mathcal{E}}(x) 
        ~ \subseteq ~ \capsule(x, r) 
        ~ \subseteq ~ \widehat{\mathcal{E}}^{\sqrt{2}}(x).
\end{equation}
\end{restatable}
%-----------------------------------------------------------------------

%-----------------------------------------------------------------------
%
\begin{proof}
By the definition of the local tensor per Equation~\ref{eq:tensor}, we may express the capsule as
\begin{equation} \label{eq:capsule}
  \capsule_{s_i}(x, r)
     ~ = ~ \left\{y \in \RE^d ~\bigg|~ \max_i\left(\frac{(y-x)^\Transpose\Hess D^\ell_i(y-x)}{\max\{r^2, 2\cdot D_i(x)\}}, \frac{(y-x)^\Transpose (y-x)}{\max\{r^2, 2\cdot D^\bullet_i(x)\}}\right) \leq 1\right\},
\end{equation}
where, in contrast to Equation~\ref{eq:tensor}, we replaced $v_iv_i^\Transpose$ in the second term by just the identity matrix to make the shape nicer. Recognizing the definition of both $\widehat{\mathcal{E}}$ and the capsule $\mathcal{C}$ as the intersection of $n$ subsets, it suffices to establish the following claim: for all segments $s_i$, and any $r \leq \dist(x, s_i)$, we have
 \begin{equation}\label{eq:capsule_ellipsoid_i}
     \mathcal{E}_i(x) 
 		~ \subseteq ~ \capsule_{s_i}(x, r) 
         ~ \subseteq ~ \mathcal{E}_i^{\sqrt{2}}(x).
 \end{equation}
 Observing that $D_i(x) = \dist^2(x, s_i) / 2$ and $D_i(x) \leq D^\bullet_i(x)$, we see that $r^2$ cannot dominate in either of the denominators in Equation~\ref{eq:capsule}. In addition, for any $y \in \RE^d$ we may write $y - x = \alpha v + \beta u$, where $u$ is a unit vector orthogonal to $v$. We obtain
\[
     y \in \mathcal{E}_i(x) 
             ~ \implies ~ \frac{1}{2}(y-x)^\Transpose \mathbb{H}_i(x)(y - x) 
                 ~ = ~ \frac{1}{2}\left(\frac{\beta^2}{D_i(x)} + \frac{\alpha^2}{D^\bullet_i(x)}\right) 
             ~ \leq ~ 1.
\]
From the above, the first term in Equation~\ref{eq:capsule} is at most $1$. For the second term, observe that $\beta^2/D^\bullet_i(x) \leq \beta^2/D_i(x)$. By making this substitution, we find that $\frac{1}{2}(\beta^2 + \alpha^2)/D^\bullet_i(x) = \frac{1}{2}\|y - x\|^2/D^\bullet_i(x) \leq 1$, implying the second term in Equation~\ref{eq:capsule} is at most $1$ as well. It follows that $y \in \capsule_{s_i}(x, r)$, establishing the first inclusion. For the second inclusion, observe that for all $y \in \capsule_{s_i}(x, r)$ we have $\frac{1}{2}(y - x)^\Transpose \mathbb{H}_i(q) (y - x) \leq 2$. This holds as both terms in Equation~\ref{eq:capsule} are at most $1$, and we have $(y-x)^\Transpose vv^\Transpose (y-x) \leq (y-x)^\Transpose (y-x)$ for all unit vectors $v$. Using Equation~\ref{eq:capsule_ellipsoid_i}, the proof follows by intersection over all $i \in [n]$.
\end{proof}
%-----------------------------------------------------------------------

For the purposes of distance approximation, we work with a scaled version of these capsules which we denote by the superscript $\capsule^\lambda$ for a scale factor $\lambda$. The scaled version of each capsule is the central scaling around $x$ by $\lambda$. When $\lambda < 1$, we say that the capsules are \emph{shrunken} (see Figure~\ref{capsules.fig}(c)).

Capsules enjoy a number of useful properties, similar to the Macbeath regions in the context of convex bodies; see~\cite{Bar00,AbM18}. In particular, capsules satisfy the following \emph{expansion-containment property}, which states that whenever two shrunken capsules overlap, a constant factor expansion of one contains the other. 

%-----------------------------------------------------------------------
\begin{lemma}[Expansion-Containment] \label{capsule-exp-con.lem}
Let $S$ be a set of disjoint line segments in $\RE^d$, distance parameter $r \geq 0$, and scale factor $0 < \lambda < 1$. For any $x, y \in \RE^d$, if $\capsule^\lambda(x, r) \cap \capsule^\lambda(y, r) \neq \emptyset$ then $\capsule^\lambda(y, r) \subseteq \capsule^{\alpha\lambda}(x, r)$, where $\alpha = \frac{3+\lambda}{1-\lambda}$.
\end{lemma}
%-----------------------------------------------------------------------

%-----------------------------------------------------------------------
\begin{proof}
It suffices to establish these expansion-containment properties for the individual capsules induced by each segment. That is, for any $s \in S$, we will show that $\capsule^\lambda_{s}(x, r) \cap \capsule^\lambda_{s}(y, r) \neq \emptyset$ implies that $\capsule^\lambda_{s}(y, r) \subseteq \capsule^{\alpha\lambda}_{s}(x, r)$. Assuming this is true, the proof follows by observing that:
\[
    C^\lambda(y, r) 
        ~ =         ~ \bigcap_{s \in S} C_s^\lambda(y, r) 
        ~ \subseteq ~ \bigcap_{s \in S} C_s^{\alpha\lambda}(x, r) 
        ~ =         ~ C^{\alpha\lambda}(x, r).
\]

By the capsule construction in Section~\ref{sec:capsules}, we deal with three cases. In what follows, if the relevant radius is $r$, we say that the capsule is \emph{near}. Otherwise, we say that the capsule is \emph{far}. In addition, when the closest point on the segment is an end point, the capsule takes the form of a ball. On the other hand, when the closest point is an interior point, the capsule becomes relatively elongated in the direction parallel to the segment and we simply refer to its shape as \textit{pseudo-cylinder}. To further simplify the case analysis below, we will prove that $\capsule^\lambda_{s}(x, r) \cap \capsule^\lambda_{s}(y, r) \neq \emptyset$ implies both $\capsule^\lambda_{s}(y, r) \subseteq \capsule^{\alpha\lambda}_{s}(x, r)$ and $\capsule^\lambda_{s}(x, r) \subseteq \capsule^{\alpha\lambda}_{s}(y, r)$. That way, we avoid the need to exchange the roles of $x$ and $y$ in each case.

\textbf{Case~(1): Two Balls}

\textbf{Case~(1.1) Both near $\left(c=3\right)$.}
In this case, both radii are $r$. As the shrunken balls intersect, $\dist(x, y) \leq 2\lambda r$. Furthermore, for any point $z \in \capsule^\lambda_{s}(y, r)$, we have $\dist(x, z) \leq 3\lambda r$. Hence, expanding $\capsule^\lambda_{s}(x, r)$ by a factor of $3$ suffices to contain $\capsule^\lambda_{s}(y, r)$. By symmetry, the same factor works when exchanging the roles of $x$ and $y$.

\textbf{Case~(1.2) One near and one far $\left(c=\frac{3+\lambda}{1-\lambda}\right)$.} 
We will only consider the case where $\dist(x, s) \leq \dist(y, s)$. (The other case is a bit easier, since $x$'s capsule is larger and $y$'s is smaller.) The radius of $\capsule_{s}(x, r)$ is $r$ while the radius of $\capsule_{s}(y, r)$ is $\dist(y, s) > r$. This implies that $\dist(x, s) \leq r$. Hence, by the triangle inequality and the fact the shrunken capsules contain some common point $z'$, we have
\begin{align*}
    \dist(y, s) 
        & ~ \leq ~ \dist(y, x) + \dist(x, s) 
          ~ \leq ~ \dist(y, z') + \dist(z', x) + \dist(x, s) \\
        & ~ \leq ~ \lambda\cdot\dist(y, s) + \lambda r + r.
\end{align*}
We conclude that $\dist(y, s) \leq \frac{1+\lambda}{1-\lambda} r$, implying that the radius of $\capsule^\lambda_{s}(y, r) \leq \frac{1+\lambda}{1-\lambda} \lambda r$. Because the shrunken capsules overlap, we have $\dist(y, x) \leq \lambda(\dist(y,s) + \dist(x,s)) \leq \frac{2\lambda}{1-\lambda} r$. Hence, for any $z \in \capsule^\lambda_{s}(y, r)$, we have $\dist(x, z) \leq \dist(y, x) + \frac{1+\lambda}{1-\lambda} \lambda r \leq \frac{3+\lambda}{1-\lambda} \lambda r$, and it since the radius of $\capsule^\lambda_{s}(x, r)$ is $\lambda r$, it suffices to expand $\capsule^\lambda_{s}(x, r)$ by $\frac{3+\lambda}{1-\lambda}$.

\textbf{Case~(1.3) Both far $\left(c=\frac{3+\lambda}{1-\lambda}\right)$.}
As before, we will only consider the case where $\dist(x, s) \leq \dist(y, s)$. In this case, we have $r < \dist(x,s) \leq \dist(y,s)$. However, as $\max(r, \dist(x,s)) = \dist(x,s)$, we may assume $r = \dist(x,s)$ without changing $\capsule_{s}(x, r)$ and $\capsule_{s}(y, r)$. It follows that Case(1.3) reduces to Case(1.2).

\textbf{Case~(2): Pseudo-Cylinder and Ball}
Let $e$ and $e'$ denote the endpoints of $s$ such that $\dist(x, e) \leq \dist(x, e')$.  Throughout the analysis of this case, we assume that the pseudo-cylinder is associated with the point $x$, while the ball is associated with the point $y$. That is, we have $\dist(x, s) < \dist(x, e)$, for otherwise we would be dealing with Case(1). Consequently, the radius of the restricting ball centered at $x$ is defined by its distance to the endpoint $e$, where $\dist(x, e) > r$.

Let us denote the direction parallel to $s$ as the \emph{vertical direction}, and the $(d-1)$-dimensional subspace orthogonal to it as the \emph{horizontal subspace}. We derive the expansion factors required for containment along the vertical and horizontal directions separately.

\textbf{Case~(2.1) $y$ is near $\left(c=\frac{3+\lambda}{1-\lambda}\right)$.}
For any $z \in \capsule^\lambda_{s}(y, r)$, we have $\dist(x, z) \leq \lambda\cdot\dist(x,e) + 2\lambda r \leq 3\lambda\cdot\dist(x, e)$. Hence, expanding $\capsule^\lambda_{s}(x, r)$ by a factor of $3$ suffices to cover $\capsule^\lambda_{s}(y, r)$ in the vertical direction. For the horizontal direction, we again reduce to Case(1.1) if $\dist(x, s) \leq r$ or Case(1.2) if $\dist(x, s) > r$. Overall, it suffices to expand $\capsule^\lambda_{s}(x, r)$ by $\frac{3+\lambda}{1-\lambda}$.

Let $e'' \in \{e, e'\}$ denote the closest point to $y$. By the triangle inequality, we have
\begin{align*}
    \dist(x, e'') ~ \leq ~ \dist(y, e'') + \dist(x, y) 
            & ~ \implies ~ \dist(x, e) ~ \leq ~ r + \lambda r + \lambda\cdot \dist(x, e) \\
            & ~ \implies ~ \dist(x, e) ~ \leq ~ \frac{1+\lambda}{1-\lambda} r. 
\end{align*}
Hence, for any $z \in \capsule^\lambda_{s}(x, r)$ we have $\dist(y, z) \leq \lambda r + 2\lambda\cdot\dist(x,e) \leq \frac{3+\lambda}{1-\lambda} \lambda r$, and it suffices to expand $\capsule^\lambda_{s}(y, r)$ by $\frac{3+\lambda}{1-\lambda}$.

\textbf{Case~(2.2) $y$ is far $\left(c=\frac{3+\lambda}{1-\lambda}\right)$.}
Again, letting $e'' \in \{e, e'\}$ denote the closest point to $y$, the radius of $\capsule_{s}(y, r)$ is $\dist(y, e'')$, which we denote by $r_y$. Note that $r_y \leq \dist(y, e)$ as $y$'s capsule is a ball. By the triangle inequality, we have
\begin{align*}
    \dist(y, e) ~ \leq ~ \dist(x, e) + \dist(x, y) 
        & ~ \implies ~ r_y ~ \leq ~ \dist(x, e) + \lambda\cdot \dist(x, e) + \lambda \cdot r_y \\ 
        & ~ \implies ~ r_y ~ \leq ~ \frac{1+\lambda}{1-\lambda}\dist(x,e).
\end{align*}
For any $z \in \capsule^\lambda_{s}(y, r)$, we have $\dist(x, z) \leq \lambda\cdot\dist(x,e) + 2\lambda r_y \leq \lambda \frac{3+\lambda}{1-\lambda} \cdot \dist(x, e)$. Hence, expanding $\capsule^\lambda_{s}(x, r)$ by a factor of $\frac{3+\lambda}{1-\lambda}$ suffices to cover $\capsule^\lambda_{s}(y, r)$ in the vertical direction. For the horizontal direction, we again reduce to Case(1.1) if $\dist(x, s) \leq r$ or Case(1.2) if $\dist(x, s) > r$. Overall, it suffices to expand $\capsule^\lambda_{s}(x, r)$ by $\frac{3+\lambda}{1-\lambda}$.

\medskip
Similarly, since $\dist(x, e) \leq \dist(x, e'')$, we have
\begin{align*}
    \dist(x, e'') ~ \leq ~ \dist(y, e'') + \dist(x, y) 
        & ~ \implies ~ \dist(x, e) ~ \leq ~ r_y + \lambda \cdot r_y + \lambda\cdot \dist(x, e) \\
        & ~ \implies ~ \dist(x, e) ~ \leq ~ \frac{1+\lambda}{1-\lambda} r_y.
\end{align*}
Hence, for any $z \in \capsule^\lambda_{s}(x, r)$ we have $\dist(y, z) \leq \lambda r_y + 2\lambda\cdot\dist(x,e) \leq \frac{3+\lambda}{1-\lambda} \lambda r_y$, and it suffices to expand $\capsule^\lambda_{s}(y, r)$ by $\frac{3+\lambda}{1-\lambda}$.

\textbf{Case~(3): Two Pseudo-Cylinders $\left(c=\frac{3+\lambda}{1-\lambda}\right)$} 
Assume that the radius of the restricting ball centered at $x$ is determined by its distance to $e$, while that of $y$ is determined by $e'$, where $e$ and $e'$ are not necessarily distinct, and let $r_x = \dist(x, e)$ and $r_y = \dist(y, e')$. Observe that $r_y \leq \dist(y, e)$ by definition, and $r \leq \min\{r_x, r_y\}$ as both capsules are pseudo-cylinders. By the triangle inequality, we have
\begin{align*}
    \dist(y, e) ~ \leq ~ \dist(x, e) + \dist(x, y) 
        & ~ \implies ~ r_y ~ \leq ~ r_x + \lambda r_x + \lambda r_y  \\
        & ~ \implies ~ r_y ~ \leq ~ \frac{1+\lambda}{1-\lambda} r_x.
\end{align*}
Hence, for any $z \in \capsule^\lambda_{s}(y, r)$ we have $\dist(x, z) \leq \lambda r_x + 2 \lambda r_y \leq \frac{3+\lambda}{1-\lambda} \lambda r_x$. Hence, expanding $\capsule^\lambda_{s}(x, r)$ by a factor of $\frac{3+\lambda}{1-\lambda}$ suffices to cover $\capsule^\lambda_{s}(y, r)$ in the vertical direction. For the horizontal direction, we again reduce to Case(1) depending on $\dist(x, s)$ and $\dist(y, s)$. Overall, it suffices to expand $\capsule^\lambda_{s}(x, r)$ by $\frac{3+\lambda}{1-\lambda}$. By symmetry, the same factor works when exchanging the roles of $x$ and $y$.
\end{proof}
%-----------------------------------------------------------------------

%=======================================================================
\subsection{Local Feature Size} \label{LFS.sec}
%=======================================================================

In order to use capsules for space covering, we need a principled way to select the distance parameter $r$. We define the \emph{local feature size} (\emph{LFS}) at $x \in \RE^d$ as the distance from $x$ to the \emph{second-nearest} segment:
\begin{equation}\label{LFS.eq}
    \LFS(x) 
        ~ = ~ \min_{i,j \in \binom{n}{2}}\max\{\dist(x, s_i), \dist(x, s_j)\},
\end{equation}
where we assume $n \geq 2$. The local feature size cannot change too rapidly, as shown in the next lemma.

%-----------------------------------------------------------------------
\begin{lemma} \label{Lipschitz.lem}
$\LFS$ is 1-Lipschitz continuous, that is $|\LFS(x) - \LFS(y)| \leq \|x - y\|$.
\end{lemma}
%-----------------------------------------------------------------------

%-----------------------------------------------------------------------
\begin{proof} 
By applying the triangle inequality (setwise), for any segment $s$ we have $\dist(x, s) \leq \dist(x, y) + \dist(y, s)$. By the definition of $\LFS$, for all $x, y \in \RE^d$, we have
\begin{align*}
    \LFS(x) 
        & ~ =    ~ \min_{i,j \in \binom{n}{2}}\max\{\dist(x, s_i), \dist(x, s_j)\} \\
        & ~ \leq ~ \dist(x, y) + \min_{i,j \in \binom{n}{2}}\max\{\dist(y, s_i), \dist(y, s_j)\} 
          ~ =    ~ \dist(x, y) + \LFS(y).
\end{align*}
Repeating this with $x$ and $y$ swapped yields the desired conclusion.
\end{proof}
%-----------------------------------------------------------------------

The following lemma further quantifies the sensitivity of capsules to the distances to the set of line segments. In particular, all points within a shrunken capsule have comparable local feature size.

%-----------------------------------------------------------------------
\begin{lemma} \label{LFS-passport.lem}
For all $z \in \capsule^\lambda(x, \LFS(x))$, where $0 < \lambda < 1$, $\LFS(z) \in [1 - \lambda, 1 + \lambda]\cdot\LFS(x)$.
\end{lemma}
%-----------------------------------------------------------------------

%-----------------------------------------------------------------------
\begin{proof}
Let $r = \LFS(x)$, and denote by $s_1$ and $s_2$ the nearest and second-nearest segments to $x$, respectively. Observing that $r = \dist(x, s_2)$, the interior of $\capsule(x, r)$ cannot intersect any segment except for $s_1$. To obtain the lower bound, we bound the distance from any $z \in \capsule^\lambda(x, r)$ to $\bd \SP \capsule(x, r)$.
 
Recalling the construction of capsules, for any segment $s$, as an infinite cylinder restricted within a ball, we define the \emph{spine} of the capsule at $x$ with respect to $s$ as the projection of $s$ onto the axis of the cylinder intersected with the capsule, or just $x$ if the capsule is a ball; see Figure~\ref{spine.fig}.
 
%-----------------------------------------------------------------------
\begin{figure}[htbp]
  \centerline{\includegraphics[scale=0.6]{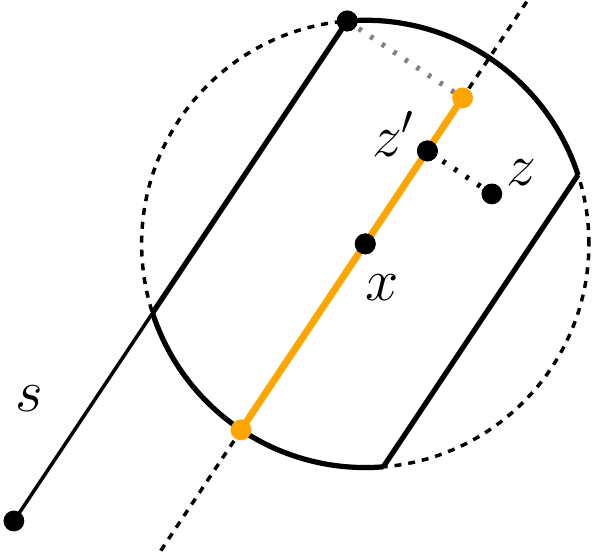}}
  \caption{The spine construction used in Lemma~\ref{LFS-passport.lem}.}
  \label{spine.fig}
\end{figure}
%-----------------------------------------------------------------------

For any $z \in \capsule_{s_2}(x, r)$, define $z'$ as the projection of $z$ onto the spine. By construction, $\dist(z, z') \leq r$ and $\dist(z', s_2) = r$, for all $z \in \capsule_{s_2}(x, r)$ and $z'$ on the spine. Upon shrinking, we obtain $\dist(z, z') \leq \lambda r$. In addition, $\dist(z, \bd \SP \capsule(x, r)) \geq (1-\lambda)r$. This lower bound is obvious for the shrunken cylindrical shell. For the spherical caps, we note that they are at least as far from $x$ as the cylindrical shell, so the spherical caps of the shrunken capsule are displaced by at least the same amount as the shrunken cylindrical shell. Since only $s_1$ may be closer to $z$ than $\bd \SP \capsule(x, r)$, we have $\LFS(z) \geq (1-\lambda)r$.
 
For the upper bound, we consider $\capsule_{s_1}(x, r)$ in addition to $\capsule_{s_2}(x, r)$. For any $z \in \capsule_{s_1}(x, r) \cap \capsule_{s_2}(x, r)$, let $z'$ and $z''$ denote the projections of $z$ onto the spines of $\capsule_{s_1}(x, r)$ and $\capsule_{s_2}(x, r)$, respectively. By the above derivations $\dist(z, s_2) \leq \dist(z, z'') + \dist(z'', s_2) \leq (1+\lambda)r$. We also have $\dist(z', s_1) = \dist(x, s_1) \leq \LFS(x) = r$, implying $\dist(z, s_1) \leq \dist(z, z') + \dist(z', s_1) \leq (1+\lambda)r$. It follows that, 
 \[
    \LFS(z) 
        ~ \leq ~ \max\{\dist(z, s_1), \dist(z, s_2)\} 
        ~ \leq ~ (1+\lambda)r,
\]
as desired.
\end{proof}
%-----------------------------------------------------------------------

From the above lemma, it is easy to obtain an approximation of nearest-neighbor distances using the segment closest to the center point of the capsule as a representative. This qualifies capsules to serve as cells for a type of approximate Voronoi diagram (AVD) data structure~\cite{Har01,AVD-JACM}.

%-----------------------------------------------------------------------
\begin{lemma} \label{capsule-approx.lem}
For any $z \in \capsule^\lambda(x, \LFS(x))$ with $0 < \lambda < 1$, $\dist(z, s_x) \leq \frac{1+\lambda}{1-\lambda} \dist(z, s_z)$, where $s_x$ is the closest segment to $x$ and $s_z$ is the closest segment to $z$.
\end{lemma}
%-----------------------------------------------------------------------

%-----------------------------------------------------------------------
\begin{proof}
Assume $s_z$ is distinct from $s_x$, for the assertion holds trivially otherwise. In the same notation we used to prove Lemma~\ref{LFS-passport.lem}, let $r = \LFS(x)$, and denote by $s_x$ and $s'_x$ the nearest and second-nearest segments to $x$, respectively. Observing that $r = \dist(x, s'_x)$, the interior of $\capsule(x, r)$ cannot intersect any segment except for $s_x$. As seen in the proof of Lemma~\ref{LFS-passport.lem}, we have $\dist(z, s_z) \geq \dist(z, \bd \SP \capsule(x, r)) \geq (1-\lambda)r$ while $\dist(z, s_x) \leq (1+\lambda)r$.
\end{proof}
%-----------------------------------------------------------------------

As an immediate corollary, we have the following.

%-----------------------------------------------------------------------
\begin{corollary} \label{capsule-approx.cor}
For $0 < \eps \leq 1$ and any $z \in \capsule^\lambda(x, \LFS(x))$ with $0 < \lambda \leq \frac{\eps}{3}$, the nearest neighbor of $x$ is a $(1+\eps)$-approximate nearest neighbor of $z$.
\end{corollary}
%-----------------------------------------------------------------------

Capsules enjoy a number of properties (described later in this section) that make them suitable for forming hierarchical covers of space. However, as the intersection of $n$ cylinders and/or balls, capsules can have high combinatorial complexity. For this reason, we use their John ellipsoids in our data structure instead, as in related data structures based on Macbeath regions~\cite{AbM18,AFM18b}. For $x \in \RE^d$ and positive scalars $\lambda$ and $r$, define $E^\lambda(x, r)$ as the maximum volume ellipsoid enclosed within $\capsule^\lambda(x, r)$. By John's Theorem \cite{Joh48}, $E^\lambda(x, r) \subseteq \capsule^\lambda(x, r) \subseteq E^{\lambda\sqrt{d}}(x, r)$. Hence, up to constant factors, these \emph{capsule ellipsoids} can serve as low-complexity proxies for capsule. Our construction makes use of two particular constant scale factors independent of $\eps$, $0 < \lambda'' < \lambda' < 1$. For any $x$ and $r$, define $E''(x, r) = E^{\lambda''}(x, r)$ and $E'(x, r) = E^{\lambda'}(x, r)$. 

There is some flexibility in how these scale factors are chosen. A convenient choice is to select $\lambda' = \inv{2}$ (or any constant fraction bounded away from $1$) and select $\lambda''$ to be $\lambda'\sqrt{d}/\alpha$, where $\alpha$ is the scale factor from expansion-containment. This choice implies that any maximal packing of $E''$ ellipsoids can be converted to a cover by replacing each with its $E'$ ellipsoid. 

%=======================================================================
\subsection{Net-Like Properties for Capsules} \label{properties.sec}
%=======================================================================

In this section, present a number of properties of capsules, demonstrating that they possess similar properties to nets, which arise in the study of metric spaces~\cite{KrL04,KrL05,HaM06}. While we prove these results for capsules, they all hold for capsule ellipsoids, subject to an adjustment of constant factors. 

Our first result is a utility that relates two methods for growing capsules, first by expanding the distance parameter and second by applying a scale factor.

%-----------------------------------------------------------------------
\begin{lemma}\label{radius-scaling.lem}
For any $\gamma \geq 1$, $\capsule(x, \gamma r) \subseteq \capsule^\gamma(x, r)$, and for $0 < \lambda \leq 1$, $\capsule^\lambda(x, r) \subseteq \capsule(x, \lambda r)$.
\end{lemma}
%-----------------------------------------------------------------------

%-----------------------------------------------------------------------
\begin{proof}
For each segment $s$, the radii used in $\capsule_{s}(x, r)$ are of the form $\max(r, \dist(x, s))$. Clearly, $\max(\gamma r, \dist(x, s)) \leq \gamma \cdot \max(r, \dist(x, s))$. The other inequality is similar.
\end{proof}
%-----------------------------------------------------------------------

The next lemma bounds the growth in the volume of capsules upon scaling. The first part follows directly from the fact that capsules are full dimensional and convex, and the second part follows from this in combination with Lemma~\ref{radius-scaling.lem}.

%-----------------------------------------------------------------------
\begin{lemma}\label{capsule-vol-bound.lem}
For any set of disjoint line segments $S\subseteq\RE^d$ and an arbitrary $x \in \RE^d$:
\begin{enumerate}
\vspace{5pt}\setlength{\itemsep}{2pt}%
\item[$(i)$] For $\lambda > 0$, $\vol(\capsule^\lambda(x, r)) ~ = ~ \lambda^d\cdot\vol(\capsule(x, r))$.

\item[$(ii)$] For $\beta \geq 1$, $\vol(\capsule(x, \beta r)) ~\leq~ \beta^d\cdot\vol(\capsule(x, r))$.
\end{enumerate}
\end{lemma}
%-----------------------------------------------------------------------

Next, we derive a packing bound on the number of pairwise interior-disjoint capsules that may fit within a larger capsule.

%-----------------------------------------------------------------------
\begin{lemma}\label{capsule-deg-bound.lem}
Given a set of disjoint line segments $S\subseteq\RE^d$, $r \geq 0$, and two constant scale factors $0 < \lambda_p < \lambda_c < 1$, let $Y \subset \RE^d$ denote a set of points such that the associated capsules $\capsule^{\lambda_p}(y, r)$, with $y \in Y$, are disjoint. Then, for any $x \in \RE^d$ and $\beta \geq 1$, the number of capsules in $R_Y = \{\capsule^{\lambda_c}(y, r) \mid y \in Y\}$ that intersect $\capsule^{\lambda_c}(x, \beta r)$ is $O(\beta^d)$.
\end{lemma}
%-----------------------------------------------------------------------

%-----------------------------------------------------------------------
\begin{proof}
Fix a $y \in Y$ such that $\capsule^{\lambda_c}(x, \beta r) \cap \capsule^{\lambda_c}(y, r) \neq \emptyset$. As $\capsule^{\lambda_c}(y, r) \subseteq \capsule^{\lambda_c}(y, \beta r)$, we also have $\capsule^{\lambda_c}(x, \beta r) \cap \capsule^{\lambda_c}(y, \beta r) \neq \emptyset$. Applying Lemma~\ref{capsule-exp-con.lem} (with the roles of $x$ and $y$ swapped), we obtain $\capsule^{\lambda_c}(x, \beta r) \subseteq \capsule^{\alpha\lambda_c}(y, \beta r)$, where $\alpha = \frac{3+\lambda_c}{1-\lambda_c} > 1$. Lemma~\ref{capsule-vol-bound.lem} yields
\begin{align*}
    \vol(\capsule^{\lambda_p}(y, r))
        & ~ = ~ \left( \frac{\lambda_p}{\lambda_c} \right)^{\kern-2pt d} \vol(\capsule^{\lambda_c}(y, r)) 
          ~ \geq ~ \left( \frac{\lambda_p}{\lambda_c\beta} \right)^{\kern-2pt d} \vol(\capsule^{\lambda_c}(y, \beta r)) \\
        & ~ = ~ \left( \frac{\lambda_p}{\lambda_c\alpha\beta} \right)^{\kern-2pt d} \vol(\capsule^{\alpha\lambda_c}(y, \beta r))
          ~ \geq ~ \left( \frac{\lambda_p}{\lambda_c\alpha\beta} \right)^{\kern-2pt d} \vol(\capsule^{\lambda_c}(x, \beta r)).
\end{align*}
By packing, the number of capsules of $R_Y$ intersecting $\capsule^{\lambda_c}(x, \beta r)$ is $O\big(\big(\frac{\lambda_c\alpha\beta}{\lambda_p}\big)^d\big)$. The result follows since $\lambda_c$, $\lambda_p$, and $\alpha$ are all constants.
\end{proof}
%-----------------------------------------------------------------------

Turning our attention to radius assignment through the local feature size $\LFS$, we show that expansion-containment still holds.

%-----------------------------------------------------------------------
\begin{lemma} \label{exp-con-lfs.lem}
Given two points $x, y \in \RE^d$ and $0 < \lambda < 1$, if $\capsule^\lambda(x, \LFS(x)) \cap \capsule^\lambda(y, \LFS(y)) \neq \emptyset$, then $\capsule^\lambda(y, \LFS(y)) \subseteq \capsule^{\beta\lambda}(x, \LFS(x))$ for a constant $\beta = \frac{(3+\lambda)(1+\lambda)}{(1-\lambda)^2}$.
\end{lemma}
%-----------------------------------------------------------------------

%-----------------------------------------------------------------------
\begin{proof}
Assuming first that $\LFS(y) \leq \LFS(x)$, $\capsule^\lambda(y, \LFS(y)) \subseteq \capsule^\lambda(y, \LFS(x))$ and thus
\[
    \capsule^\lambda(x, \LFS(x)) \cap \capsule^\lambda(y, \LFS(y)) \neq \emptyset 
        ~ \implies ~ \capsule^\lambda(x, \LFS(x)) \cap \capsule^\lambda(y, \LFS(x)) \neq \emptyset.
\]
Applying Lemma~\ref{capsule-exp-con.lem} with $r = \LFS(x)$, implies that $\capsule^\lambda(y, \LFS(y)) \subseteq \capsule^{\alpha\lambda}(x, \LFS(x))$.

Otherwise, $\LFS(x) \leq \LFS(y)$, and by applying Lemma~\ref{LFS-passport.lem} twice with any $z \in \capsule^\lambda(x, \LFS(x)) \cap \capsule^\lambda(y, \LFS(y))$, we obtain
\[
    \LFS(y) 
        ~ \leq ~ \frac{1}{1-\lambda}\LFS(z) 
        ~ \leq ~ \frac{1+\lambda}{1-\lambda}\LFS(x).
\]
By Lemma~\ref{radius-scaling.lem}, $\capsule^\lambda(x, \gamma \SP \LFS(x)) \subseteq \capsule^{\gamma\lambda}(x, \LFS(x))$, where $\gamma = \frac{1+\lambda}{1-\lambda} > 1$. Therefore
\[
    \capsule^\lambda(x, \LFS(x)) \cap \capsule^\lambda(y, \LFS(y)) \neq \emptyset 
        ~ \implies ~ \capsule^\lambda(x, \gamma\cdot\LFS(x)) \cap \capsule^\lambda(y, \gamma\cdot\LFS(x)) \neq \emptyset.
\]
Applying Lemma~\ref{capsule-exp-con.lem} with $r = \gamma \SP \LFS(x)$, implies that $\capsule^\lambda(y, \LFS(y)) \subseteq \capsule^{\alpha\gamma\lambda}(x, \LFS(x))$.
\end{proof}
%-----------------------------------------------------------------------

%=======================================================================
\section{The ANN Data Structure} \label{data-struct.sec}
%=======================================================================

In this section, we apply the results of the previous section to present our data structure for answering $\eps$-ANN queries. Again, $S = \{s_1, \ldots, s_n\}$ is a set of $n$ disjoint line segments in $\RE^d$, and $\eps > 0$ is the approximation parameter. Let $B(S)$ be a minimum volume Euclidean ball that contains $S$, and let $B^+(S)$ denote a concentric expansion of $B(S)$ about its center by a factor of $1 + 2/\eps$. It is easy to see that if the query point $q$ lies outside of $B^+(S)$, any segment may be reported as an $\eps$-ANN of $q$. Thus, for the rest of the construction, we focus on query points lying within $B^+(S)$. Let $x_0$ and $r^+$ denote the center and radius of this ball, respectively. Clearly, $r^+ = \Theta(\diam(S)/\eps)$. Let $\delta(S)$ denote the minimum distance between any two segments of $S$. Observe that for every point $x \in B^+(S)$, its local feature size, $\LFS(x)$, is at least $\delta(S)/2$.

Here is a high-level overview of the data structure. It consists of a rooted directed acyclic graph (DAG), which is based on covering $B^+(S)$ with a hierarchy of capsule ellipsoids of exponentially diminishing scales. The DAG is organized in levels, with a single root node at level zero whose associated capsule ellipsoid contains $B^+(S)$. For $i \geq 0$, the capsule ellipsoids associated with the nodes of level $i$ employ the distance parameter $r_i = r^+ / 2^i$, and thus successive levels are more refined. Each level of the DAG will be associated with a collection of capsule ellipsoids that cover $B^+(S)$. In particular, each node at level $i$ of the DAG stores a point $x \in B^+(S)$, and the associated capsule ellipsoid, denoted $E'(x)$, centered at this point is defined to be the shrunken ellipsoid $E'(x, r_i)$ with respect to $S$. Each node at level $i$ will either be declared to be a leaf, or it will be linked to those nodes at level $i+1$ whose capsule ellipsoids it overlaps, which we call its \emph{children}. (We will show that the out-degree of any node is a constant.) 

We continue this refinement process until $r_i \leq \LFS(x)$. The resulting terminal nodes are called \emph{basic leaves}. We then further refine these ellipsoids through scaling. In particular, for $j = 1, 2, \ldots$, we cover $E'(x, r_i)$ by ellipsoids of the form ${E'}^{1/2^{j}}(y, r_i)$ for $y \in E'(x, r_i)$. We stop when $1/2^j \leq \frac{\eps(1-\lambda')}{3 \lambda'}$, where $\lambda'$ is the constant scale factor used in the definition the $E'$ ellipsoids. Each of the resulting ellipsoids, called a \emph{final leaf}, stores the segment of $S$ that is closest to its center as its \emph{representative}.

Queries are answered by a simple descent through this DAG. Assuming that the query point $q$ lies within $B^+(S)$, we descend level by level through the DAG. On arriving at a non-leaf node, we inspect the ellipsoids of its children on level $i+1$. Their associated capsules cover $E'(x)$, and the search continues with any one of these children whose associated ellipsoid contains $q$. When the search arrives at a final leaf, the associated representative segment is returned as the answer to the query.

The DAG is constructed in a top-down manner, starting with the root of the DAG. The root capsule is $E'(x_0)$, where $x_0$ is the center of $B^+(S)$. We assert that this covers $B^+(S)$. To see this, observe that by definition, $E'(x,r)$ contains the ball of radius $r$ centered at $x$, and hence the same holds for $E'(x,r)$. For $i = 1, 2, \ldots$, let $U_i$ denote the portion of $B^+(S)$ that is not covered by any of the leaves of the structure from prior levels. For any $x \in \RE^d$, define $E''(x) = E''(x, r_i)$. Let $X_i$ be any maximal set of points $x$ within $U_i$ such that the associated ellipsoids $E''(x)$ are pairwise disjoint. It follows from maximality and expansion-containment that the union of the expanded ellipsoids $E'(x)$ for $x \in X_i$ covers $U_i$. We create a node at level $i$ of the DAG for each point $x \in X_i$, and we link each such node as a child of any non-leaf node from the previous level whose $E'$ capsule ellipsoid (computed with respect to level $i-1$) it overlaps.

In the remainder of this section, we analyze the correctness, query time and storage requirements of this data structure. Our first two lemmas establish correctness and bound the depth of the data structure. The following shows that queries are answered correctly.

%-----------------------------------------------------------------------
\begin{lemma} \label{ann-correctness.lem}
Given a set $S$ of line segments, the above search algorithm returns an $\eps$-ANN among the segments of $S$ for any query point $q \in B^+(S)$.
\end{lemma}
%-----------------------------------------------------------------------

%-----------------------------------------------------------------------
\begin{proof}
Without loss of generality, we may assume that $\eps \leq 1$. Recall that a node centered at a point $x$ is declared to be a basic leaf when the associated distance parameter $r_i$ is less than $\LFS(x)$, and consider any final leaf ${E'}^{1/2^{j}}(y, r_i)$, for $y \in E'(x, r_i)$. By construction, $1/2^j \leq \frac{\eps(1-\lambda')}{3 \lambda'}$, where $\lambda'$ is the scale factor used in the definition of $E'$. Therefore,
\[
    {E'}^{1/2^{j}}(y, r_i)
        ~ \subseteq ~ {E'}^{\eps(1-\lambda')/3 \lambda'}(y, \LFS(x))
        ~ =         ~ \capsule^{(1-\lambda') \eps/3}(y, \LFS(x)).
\]
Since $y \in E'(x, r_i) \subseteq \capsule^{\lambda'}(x, \LFS(x))$, Lemma~\ref{LFS-passport.lem} implies that $\LFS(y) \geq (1-\lambda') \LFS(x)$. By applying Lemma~\ref{radius-scaling.lem}, we obtain $\capsule^{1-\lambda'}(y, \LFS(x)) \subseteq \capsule(y, (1-\lambda')\LFS(x))$. Given our definition of final leaf, we have
\[
    \capsule^{(1-\lambda') \eps/3}(y, \LFS(x)) 
        ~ \subseteq ~ \capsule^{\eps/3}(y, (1-\lambda') \LFS(x)) 
        ~ \subseteq ~ \capsule^{\eps/3}(y, \LFS(y)).
\]
Finally, Corollary~\ref{capsule-approx.cor} shows that the nearest neighbor of $y$ is an $\eps$-ANN to any point within  $\capsule^{\eps/3}(y, \LFS(y))$, implying that queries are answered correctly.
\end{proof}
%-----------------------------------------------------------------------

Next we analyze the depth of the DAG.

%-----------------------------------------------------------------------
\begin{lemma} \label{root-level.lem}
Given a set $S$ of $n$ disjoint segments in $\RE^d$ with spread $\Delta$, the $\eps$-AVD structure described above has $O(\log (\max(n,\Delta) / \eps))$ levels. 
\end{lemma}
%-----------------------------------------------------------------------

%-----------------------------------------------------------------------
\begin{proof}
Recall that $B(S)$ is the minimum Euclidean ball containing $S$ and $B^+(S)$ is its expansion by $1 + \frac{2}{\eps}$, and $r^+$ is its radius. Clearly, $r^+ = \Theta(\diam(S)/\eps)$. Letting $\delta_{\min}$ denote the minimum distance between any pair of segments, $\Delta = \diam(S)/\delta_{\min}$. Clearly, for any point $x$, $\LFS(x) \geq \delta_{\min}$. The refinement process terminates at the basic-leaf level when the scale falls below $\LFS(x) \geq \delta_{\min}$. It then continues for an additional $O(\log \inv{\eps})$ levels until reaching the final-leaf level. Since the scale decreases by a factor of $2$ with each level of the data structure, the total number of levels is 
\[
    O\left( \log \frac{r^+}{\delta_{\min}} + \log \inv{\eps} \right) 
        ~ = ~ O\left( \log \frac{\diam(S)}{\eps \SP \delta_{\min}}  + \log \inv{\eps} \right) 
        ~ = ~ O\left( \log \frac{\Delta}{\eps} \right),
\]
as desired. Recall that the spread of a set of segments in $\RE^d$ grows at least polynomially with $n$, therefore $\log n$ is $O(\log \Delta)$.
\end{proof}
%-----------------------------------------------------------------------

Our next result bounds the number of children for each node.

%-----------------------------------------------------------------------
\begin{lemma} \label{ann-child-bound.lem}
Each non-leaf node of the data structure has $O(1)$ children.
\end{lemma}
%-----------------------------------------------------------------------

%-----------------------------------------------------------------------
\begin{proof}
Consider a node at some level $i$ centered at a point $x$. Let $r = r_{i+1} = r_i/2$. This node's children consist of the nodes $y$ of level $i+1$ whose ellipsoid $E'(y, r_{i+1})$ overlaps $E'(x,r_i)$. By construction, all such points $y$ come from a set $Y$ whose ellipsoids $E''(y, r_{i+1})$ are disjoint. By applying Lemma~\ref{capsule-deg-bound.lem} in the elliptical setting, the number of such overlapping ellipsoids is $O(2^d) = O(1)$, given our assumption that the dimension $d$ is fixed.
\end{proof}
%-----------------------------------------------------------------------

Since each node of the DAG has constant degree, it follows that the overall query time is proportional to DAG's height, which is $O(\log (\Delta / \eps))$. Finally, we bound the total space used by the data structure.

%-----------------------------------------------------------------------
\begin{lemma} \label{ann-space-bound.lem}
Given a set $S$ of $n$ line segments in $\RE^d$, the total storage required by the $\eps$-AVD is $O((n^2 / \eps^d) \log\frac{\Delta}{\eps})$.
\end{lemma}
%-----------------------------------------------------------------------

%-----------------------------------------------------------------------
\begin{proof}
We distinguish between two transitions within the DAG structure. When the distance parameter $r_i$ of a node $x$ first falls below $\LFS(x)$, we say that this is a basic leaf, and when it falls below $\frac{\eps}{3} \LFS(x)$ (the actual termination condition), we say it is a final leaf. Lemma~\ref{charging.lem} (presented later in Section~\ref{sec:storage}) states that the number of capsules at the basic leaf level that are charged to any pair of segments is $O(\log\frac{\Delta}{\eps})$. Therefore, the total number basic leaves is $O(n^2 \log\frac{\Delta}{\eps})$. Observe that all of the points lying within $E'(x,\LFS(x))$ share the same local-feature size values up to constant factors, and therefore, all the descendants of this node lie within the next $O(\log \inv{\eps})$ levels of the structure. They are all similarly shaped (up to constant factors), but their sizes are smaller by a factor of at least $\frac{\eps}{3}$. By the disjointness of the shrunken $E''$ capsule ellipsoids, it follows that the number of descendants is $O(1/\eps^d)$. Therefore, the total number of nodes in the DAG is $O((n^2/\eps^d) \log\frac{\Delta}{\eps})$.
\end{proof}
%-----------------------------------------------------------------------

By combining the results of the previous lemmas, we obtain Theorem~\ref{ann-main.thm}.

%=======================================================================
\section{Storage Bounds for Capsules} \label{sec:storage}
%=======================================================================

Recall the definition of capsules from Section~\ref{sec:capsules}. The infinite cylindrical component in the definition of $\capsule(x, r)$ induced by a particular segment $s_i \in S$ will be denoted $\Cyl_i(x, r)$ and its radius will be denoted by $t_i(x)$.

For a given point $x \in \RE^d$ and $r > 0$, we distinguish two such cylinders. Without loss of generality, let $s_1$ denote the nearest neighbor of $x$ in $S$. Letting $v_1$ be a unit vector parallel to $s_1$ (the direction does not matter by central symmetry). Denote by $\ell_x$ the line passing through $x$ in the direction of $v_1$. Using $\ell_x$, we define the set of points $p_i$ as the intersection of $\bd \Cyl_i(x, r)$ and $\ell_x$ such that $\langle v_1, p_i - x \rangle > 0$. Again, without loss of generality, let $p_2$ denote the closest intersection point to $x$, such that $\Cyl_2(x, r)$ is the cylinder generating $p_2$. We use the two cylinders $\Cyl_1(x, r)$ and $\Cyl_2(x, r)$ to sandwich the capsule $\capsule(x, r)$ between two simple shapes providing lower and upper bounds on its volume. The radii of $\Cyl_1(x, r)$ and $\Cyl_2(x, r)$ will be denoted by $t_1$ and $t_2$, respectively. By definition, $t_1=r$. When $r$ is chosen as the local feature size $\phi(x)$ at $x$, we also have $t_2\geq r$.

The inner bounding volume $V^-(x, r)$ is defined as the double cone whose axis is $\ell_x$ and base is the $(d-1)$-dimensional disk of radius $r$ centered at $x$ orthogonal to $\ell_x$, with two apexes at $p_2$ and (symmetrically about $x$) $2x - p_2$. The outer bounding volume $V^+(x, r)$ is defined as the cylinder with $\ell_x$ as axis whose radius is $r$ and height is equal to the length of the projection of $\Cyl_1(x, r) \cap \Cyl_2(x, r)$ onto $\ell_x$.

%-----------------------------------------------------------------------
\begin{lemma}
Fix a point $x \in \RE^d$ and let $C(x, r)$ be the capsule of radius $r$ induced at $x$ by a set $S$ of $n$ line segments. Then,
\[
    V^-(x, r) 
        ~ \subseteq ~ \capsule(x, r) 
        ~ \subseteq ~ V^+(x, r), 
        \: \text{ and } \:\: 
        \frac{\vol(V^+(x, r))}{\vol(V^-(x, r))} 
        ~ \leq ~ 2 (d+1).
\]
\end{lemma}
%-----------------------------------------------------------------------

%-----------------------------------------------------------------------
\begin{proof}
The containment follows by the construction of the bounding volumes. In bounding the ratio of the two volumes, we use the same notation and assumptions as above, without loss of generality. Letting $\theta$ denote the acute angle between the two lines supporting $s_1$ and $s_2$, and $\mathbb{V}_{d-1}$ denote the volume of a unit ball in $\RE^{d-1}$, we have
\[
    \frac{\vol(V^+(x, r))}{\vol(V^-(x, r))}
        ~ \leq ~ \dfrac{ 4 \SP \mathbb{V}_{d-1} r^{d-1} t_1(x) \cdot \csc(\theta)}{\frac{2}{d+1} \mathbb{V}_{d-1} r^{d-1} t_1(x) \cdot \csc(\theta)} 
        ~ \leq ~ 2\cdot(d+1),  
\]
where the numerator is the volume of the intersection of two cylinders, and the denominator is the volume of a cone in $\RE^d$.
\end{proof}

In order to bound the number of leaf-level capsules within a ball of radius $O\left(\frac{1}{\eps}\cdot\diam(S)\right)$, we use the following charging scheme. Again, we use the simplified notation and assumptions from before. A capsule $\capsule(x, r)$ will be charged to the two line segments $s_1$ and $s_2$. For a fixed pair of segments $s_i$ and $s_j$, acting as respectively as $s_1$ and $s_2$ for the point $x$ in consideration, we may restrict attention to all center points $x$ lying in the cylinder of radius $r$ with the line supporting $s_i$ as axis; denote this cylinder by $\Cyl(s_i, r)$.

%-----------------------------------------------------------------------
\begin{lemma} \label{charging.lem}
The number of leaf level capsules charged to any pair of segments is $O(\log \frac{\Delta}{\eps})$.
\end{lemma}
%-----------------------------------------------------------------------

%-----------------------------------------------------------------------
\begin{proof}
We cover the points in $\Cyl(s_i, r)$ using a sequence of growing cylindrical intervals based on their distances to $s_j$. Define $\mathcal{I}_k(s_i, r)$ to be the set of points in $\Cyl(s_i, r)$ such that for all $x \in \mathcal{I}_k(s_i, r)$ we have $2^{k-1} r \leq \dist(x, s_j) \leq 2^k r$. In other words, $\mathcal{I}_k(s_i, r)$ is the intersection of $\Cyl(s_i, r)$ with the \emph{cylindrical shell} $\Cyl(s_j, 2^k r)\setminus\Cyl(s_j, 2^{k-1} r) $. It is easy to see that the volume of the intersection is maximized when $s_i$ intersects $s_j$ due to the symmetry of the cylinder $\Cyl(s_i, r)$ about the line supporting $s_i$. Similar to the upper bound on the volume of $\capsule(x, r)$ by that of $V^+(x, r)$, we see that $\vol(\mathcal{I}_k(s_i, r))$ is at most $4 \SP \mathbb{V}_{d-1} 2^k r^d \cdot \csc(\theta)$, where $\theta$ is the acute angle between the two lines supporting $s_i$ and $s_j$. For any capsule charged to $s_i$ and $s_j$ with center $x \in \mathcal{I}_k(s_i, r)$, we use the inner volume $V^-(x, r) \subseteq \capsule(x, r)$ to obtain a lower bound $\vol(C^\lambda(x, r)) \geq \frac{1}{d+1} \mathbb{V}_{d-1} 2^k r^d \lambda^d \cdot \csc(\theta)$. By choosing the capsule centers to have global packing-covering properties as a Delone set, it follows that there can be at most $4(d+1)/\lambda^d=O(1)$ capsules centered within $\mathcal{I}_k(s_i, r)$.

The desired bound follows by repeating the above argument over all scales $r$, a total of $\log(\Delta)$, and all $k$ with $2^k  r \leq \frac{1}{\eps}\cdot\diam(S)$, a total of $\log(1/\eps)$.
\end{proof}
%-----------------------------------------------------------------------

%=======================================================================
\section{Conclusions and Future Work}
%=======================================================================

We have presented a new AVD-based approach to answering $\eps$-segment ANN queries based on a hierarchical covering of space by ellipsoids. By elaborating on the intrinsic geometry underlying more general distance functions, our work helps pave the way to extend well-established techniques from data structure design and approximation algorithms in Euclidean and metric spaces to more general geometries. Specifically, we anticipate further progress in understanding the metric-like structure defined by the local tensors derived from the Hessians of distance functions. This should directly benefit the development of more efficient space covers, e.g., circumventing dependence on the spread and other geometric parameters. In addition, we expect polytope approximation techniques to enable better $\eps$-dependencies in the storage requirements. Put together, we leave it to future work to achieve both remaining tasks of eliminating dependence on the spread and improving low-level query processing to obtain $O(\log n/\eps)$ query times with only $O(n^2 / \eps^{d/2})$ storage, by analogy with the best known corresponding results for ANN against point sets under the Euclidean metric modulo the quadratic dependence on $n$ needed for storage in the case of segment-ANN.

%=======================================================================
% Bibliography
%=======================================================================

\bibliography{shortcuts,convex,segments}

\appendix

%=======================================================================
\section{Lower Bound} \label{lower-bound.sec}
%=======================================================================

In earlier work, Har-Peled \cite{Har01} remarked that a griddle-like construction implies a quadratic lower bound on the space of any AVD data structure for answering ANN queries among line segments; see also the related result by Aronov~\cite{Aro02}. For the sake of completeness, we present a proof that generalizes to our context of anisotropic elements. We assume that we are given a set $S$ of disjoint segments in $\RE^d$ and an error parameter $\eps > 0$. We make the assumption that the data structure operates in the \emph{AVD model}, which for us means that it consists of a cover of space by convex elements such that each element contains a representative segment, which is an $\eps$-ANN for any query point lying within this element. The storage required is the number of elements in the cover. While a usable AVD should also provide an efficient index for locating a member of the cover that contains any query point, our lower bound proof does not assume the existence of such an index.

%-----------------------------------------------------------------------
\begin{theorem}\label{lower-bound.thm}
Given any $\eps > 0$, there exists a $S$ set of $n$ line segments in $\RE^3$ such that any data structure in the AVD model for answering $\eps$-approximate nearest neighbor queries for $S$ requires $\Omega(n^2)$ space.
\end{theorem}
%-----------------------------------------------------------------------

%-----------------------------------------------------------------------
\begin{proof}
Given $\eps$, define $\delta$ to be any positive real value that is smaller than $1/2(1+\eps)$. Our construction will actually involve $2(n+1)$ line segments. The segment set $S$ consists of two sets $V$ and $H$ of segments in $\RE^3$, each of size $n+1$, that are laid out in a griddle-like manner parallel to the $(x,y)$-plane, with the segments of $H$ lying slightly above those of $V$. The ``vertical'' segments of $V = \{v_0, \ldots, v_n\}$ are parallel to the $y$-axis, where $v_i$ joins the points $(i,0,0)$ with $(i,n,0)$. The ``horizontal'' segments of $H = \{h_0, \ldots, h_n\}$ are parallel to the $x$-axis, where $h_j$ joins the points $(0,j,\delta)$ with $(n,j,\delta)$. Let the query domain $\mathcal{Q}$ be convex hull of $S$, that is, the 3-dimensional rectangle $[0,n] \times [0,n] \times [0,\delta]$.

%-----------------------------------------------------------------------
\begin{figure}[htbp]
  \centerline{\includegraphics[scale=0.40]{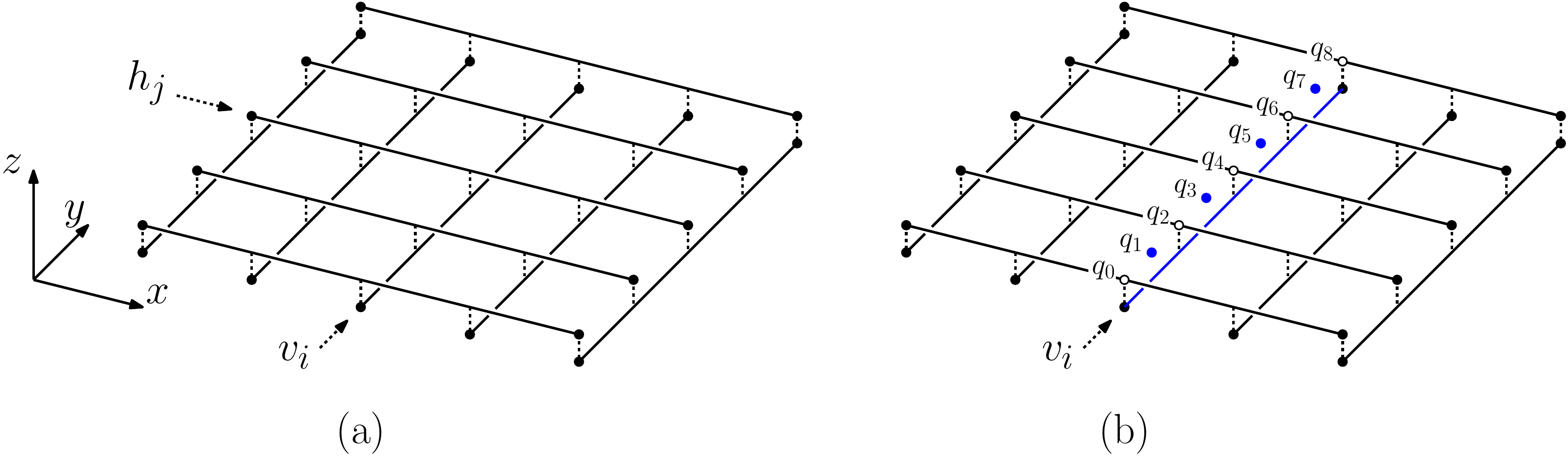}}
  \caption{Lower-bound construction.}
  \label{lower-bound.fig}
\end{figure}
%-----------------------------------------------------------------------

Intuitively, the Voronoi cell of each horizontal segment must pinch down to a small size as it approaches any point of the integer where a vertical segment crosses, but then it expands to a much large size between these grid point. This small-large oscillation along any segment implies that no convex region can be useful for approximating a segment's Voronoi cell, except along a single segment of the $n \times n$ grid. The quadratic lower bound follows as a direct consequence.

To make this intuition more formal, fix any vertical segment $v_i$ and consider a sequence of $2 n + 1$ query points placed uniformly along a parallel segment lying at distance $\delta$ above $v_i$, $\{q_0,\ldots,q_{2 n}\}$, where $q_j = \big( i, \frac{j}{2}, \delta \big)$. Observe that alternate points of this sequence $\{q_0, q_2, \ldots, q_{2 n}\}$ lie on segments of $H$, and thus these points cannot have $v_i$ has their $\eps$-nearest neighbor for any $\eps$. We assert that the remaining points $\{q_1, q_3, \ldots, q_{2 n - 1}\}$ must use $v_i$ as their representative. Each is at distance $\delta$ from $v_i$ and at distance $\inv{2}$ from the next closest segment (the horizontal segments to either side). By our choice of $\delta$, the relative error that would be committed by reporting any segment other than $v_i$ for these query points would be at least
\[
    \frac{(1/2) - \delta}{\delta}
        ~ = ~ \inv{2 \delta} - 1
        ~ > ~ (1 + \eps) - 1
        ~ = ~ \eps,
\]
which implies that $v_i$ is the only legitimate representative for these query points. Because these query points are collinear, any data structure in the AVD model must use distinct convex regions to cover the points $n$ points $\{q_1, q_3, \ldots, q_{2 n - 1}\}$, each of which has $v_i$ as its representative point. Repeating this for each of the vertical segments, it follows that the AVD needs at least $n^2$ regions. Summarizing, we have the following.
\end{proof}
%-----------------------------------------------------------------------

While the assumption that the regions in the AVD model are convex is satisfied by all existing AVD data structures for nearest neighbor searching, it is natural to wonder whether reasonable relaxations of this assumption might lead to significantly better bounds. Note that the cells of an (exact) Voronoi diagram of line segments are bounded by algebraic surfaces of constant degree. The alternating nature of the query points from our lower-bound construction implies that the quadratic lower bound would still hold even if we were to generalize the AVD model to allow (possibly non-convex) cells of constant combinatorial complexity that are bounded by algebraic surface patches of constant degree.

\end{document}